\documentclass[runningheads]{llncs}
\usepackage{booktabs} 
\usepackage{bbding}
\usepackage{amssymb}
\usepackage{fullpage}
\usepackage{amsmath,bm}
\usepackage{graphicx}
\usepackage[ruled]{algorithm2e} 

\begin{document}
\renewcommand{\algorithmcfname}{ALGORITHM}



\title{How Good Is a Two-Party Election Game?}

\author{Chuang-Chieh Lin\inst{1} \and Chi-Jen Lu\inst{1}\and Po-An Chen\inst{2}}

\authorrunning{C.-C. Lin, C.-J. Lu, and P.-A. Chen}

\institute{Institute of Information Science, Academia Sinica\\128 Academia Road, Section 2, Nankang, Taipei 11529, Taiwan
\\\email{\{josephcclin,cjlu\}@iis.sinica.edu.tw}\and
Institute of Information Management, National Chiao Tung University\\1001 University Rd, Hsinchu City 300, Taiwan \\\email{poanchen@nctu.edu.tw}}

\maketitle
\begin{abstract}
In this paper, we propose a simple and intuitive model to investigate the efficiency of the 
two-party election system, especially regarding the nomination process. Each of the two parties has its own candidates, and each of them brings utilities for the people including the supporters and non-supporters. In an election, each party nominates exactly one of its candidates to compete against the other party's. The candidate wins the election with higher odds if he or she brings more utility for all the people. 
We model such competition as a {\em two-party election game} such that each party is a player with two or more pure strategies corresponding to its potential candidates, and the payoff of each party is a mixed utility from a selected pair of competing candidates. 

By looking into the three models, namely, the linear link, Bradley-Terry, and the softmax models, which differ in how to formulate a candidate's winning odds against the competing candidate, we show that the two-party election game may neither have any pure Nash equilibrium nor a bounded price of anarchy. 
Nevertheless, by considering the conventional 
{\em egoism}, which states that any candidate benefits his/her party's supporters more than any candidate from the competing party does, we prove that the two-party election game in both the linear link model and the softmax model always has pure Nash equilibria, and furthermore, the price of anarchy is constantly bounded. 
\end{abstract}

\keywords{Two-party election game; Nash equilibrium; price of anarchy; egoism}

\section{Introduction}
\label{sec:intro}


Most of the literature on voting theory is devoted to studying the voters' behavior on a micro level and equilibria under various election rules, e.g., majority, plurality, Borda count, etc.
As we know, voters are strategic, having different preferences for the candidates. Their behaviors are sometimes complex and dependent on the voting procedure which involves the ballot structure and the allocation rule. 
Due to many factors involved, 
it is thus difficult for any micro-level model to capture the overall efficiency of a political system consisting of two or multiple competing parties.

Modern democracy by and large runs on political parties nominating their ``best" candidates to compete in elections by appealing mainly to their supporters. 
Not much has been said about how the parties form their nomination strategies in terms of candidate selection in an election, and how this process stabilizes and affects the benefits of the people in the end. 
Usually, a candidate nominated by a party responds and caters more to the needs of his/her party's supporters and less to the needs of the non-supporters. This subsequently determines different benefits/utilities that a candidate could bring to the supporters and to the non-supporters.\footnote{Note that a voter who obtains the same amount of benefits from any party's candidate can be considered as a supporter of an arbitrary party.} Furthermore, by Duverger's law suggesting that the plurality voting favors the two-party system~\cite{Duv54}, we consider a democracy consisting of two parties. Such two-party or nearly two-party system exists in many democratic countries such as the United States, the United Kingdom, Taiwan, etc. 
 
In this paper, we consider a simplified macro model instead. We formalize the political competition between two parties as a \emph{two-player election game}, in which each party as a player has candidates as its strategies and its payoff is the expected utility for its supporters considering that their selected candidate may win or lose. The odds of winning an election for a candidate nominated by a party over another candidate nominated by the competing party could arguably be related to two factors - the total benefits that he/she brings to the whole society, including the supporters and non-supporters, and those that his/her competitor brings. 

Based on the proposed simple and intuitive models, we would like to know how the two-party system benefits people by asking whether a two-party system will stabilize, i.e., having an equilibrium, and in particular whether the competition between the two parties in such a system always benefits the people optimally. If not, then how good/bad can the selected pair of competing candidates be at equilibrium, compared with the optimal pair? Specifically, we answer the former question by proving the existence of a Nash equilibrium, and analyze the latter in terms of the price of anarchy~\cite{KP09} which measures the efficiency of an equilibrium of a game. By considering the conventional 
{\em egoism}, which states that any candidate benefits his/her party's supporters more than any candidate from the competing party does, we prove that the two-party election game in some models always has pure Nash equilibria, and the price of anarchy is constantly bounded. This shows that in some sense the game between two parties in candidate nomination for an election benefits the people with a social welfare at least constantly far from the optimum.


\subsection{Our Results}

There are issues concerning previous work on political competition, which we list in the following. 
\begin{itemize}
	\item Most studies on political competitions consider policies or political positions on 
	a unidimensional metric space (e.g., a closed interval on the real line), except for very few studies 
	 (e.g., L\'{o}pez et al.~\cite{LRL2007} considers political positions as points in a two-dimensional plane). 
	 However, multiple policy dimensions are sometimes necessary to capture voter preferences. Then, what a 
	 reasonable dimension of the policy space (or candidate pool) is may not be clear. 
	 Besides, such a game of political competition may not have any pure Nash equilibrium when the 
	 policy space is multi-dimensional~\cite{Dug16}. 
	\item What kind of voter preference (e.g., proximity in a metric space) on the policy space is reasonable? 
	Besides, voter preferences may not be identical and may even be dependent. 
	\item In practical cases, the policy space is finite and believed to be countable. 
	For example, it is believed to be in proportion to the number of candidates. 
	\item What kind of voting procedures (e.g., Plurality Voting, Borda Count, Negative Voting, 
	Approval Voting) should we apply? Moreover, different voting procedures may incur 
	different strategic behaviors of voters.
\end{itemize}

In this paper, we bypass the above issues by modeling the political competition as a two-player game, 
which is called a two-party election game, and analyze the efficiency of the game on a macro scale. Specifically, the two parties are modeled as two strategic players, and each player's strategies correspond to its candidates, which are countably many and finite. 
Each candidate of a party brings utilities for the supporters of its party and also for the supporters of 
the competing party. 
A candidate wins the competition with higher odds if one brings more utility for all the voters. 
Each player's payoff is the expected utility its supporters get from the two parties. 
Moreover, we conjecture that the winning probability of a candidate is a function of the total utility 
it can bring. By considering three kinds of functions modeling the winning probability, we investigate if 
such a political competition between two parties always has a pure Nash equilibrium. Moreover, if 
the answer is affirmative, we analyze the price of anarchy of the game. 

The setting of our model, which is inspired by dueling bandits~\cite{AJK2014,YBKJ2012}, 
is simple and intuitive. We prove that the two-party election game may neither have any pure 
Nash equilibrium nor bounded price of anarchy. However, under the standard assumption of  
{\em egoism}, which states that any candidate benefits its supporters 
more than those from the opposite party, we prove that the two-party election game in both 
the linear link model and the softmax model always has pure Nash equilibria, and furthermore, 
the price of anarchy is bounded by a constant. Our findings suggest that, under the egoism 
assumption, the two-party system is stable and efficient for the people. Note that we are modeling in this paper the decision making about the candidate choices of a two-party election system, depending on the utility estimates for the supporters and non-supports. This is less about the voters' behavior or how the voters are modeled. We focus on the candidate choice instead of voting itself, which is the focus in most of the voting theory. 

\subsection{Related work}
\label{subsec:related_work}

\paragraph{Explanation of Duverger's law.} 
There have been studies that formalize Duverger's law. It can be explained either by the strategic 
behavior of the voters~\cite{Fed92,Fey97,Del2013,MW93,Pal89} or that of the 
candidates~\cite{Pal84,Web92,Cal05,CW07}. The latter, which is more relevant to our work, 
considers models where two established candidates face a potential entrant, i.e., 
the third candidates, as a threat, such that the two established candidates choose policies 
(i.e., points in the policy space) to preempt successful entry. Under mild assumptions on voter
preferences, Dellis~\cite{Del2013} explained why a two-party system emerges under plurality voting 
and other voting procedures permitting truncated ballots. 

\paragraph{Distortion.}
The notion of distortion, introduced by Procaccia and Rosenschein~\cite{PR2006}, resembles 
the notion of price of anarchy, although the latter is used in games of strategic players. 
The distortion measures the rate of social welfare decrease (or social cost increase) 
when voting is used. The relevant work includes the studies which focus on cardinal preferences 
of voters (e.g.,~\cite{CP2011,PR2006}) and those which focuses on metric preferences 
of voters (e.g.,~\cite{ABP2015,AP2016,CDK2017}). Caragiannis and Procaccia~\cite{CP2011} 
discussed the distortion when each voter's vote receives an embedding, which maps 
the preference to the output vote. Cheng et al.~\cite{CDK2017} restricted their attention 
to the distribution of voters and candidates and justified that the expected distortion is 
small when candidates are drawn from the same distribution of voters.

\paragraph{On equilibria of political competitions.} 
Most of the works on equilibria of a political competition are mainly based on 
{\em Spatial Theory of Voting} (e.g.,~\cite{Dow57,LRL2007,Pal84,Web92}, which 
was initialized by Downs~\cite{Dow57} and can be traced back to Hotelling's work~\cite{Hot29}). 
In the Hotelling-Downs model, there are two established parties facing voters with 
symmetric single-peaked preferences over a unidimensional metric space. Each party 
chooses a political position (or, a policy) that is as close as possible to the greatest number of 
voter preferences. The Spatial Theory of Voting states that when the policy (i.e., candidate) 
space is unidimensional, voter preferences are single-peaked and two parties compete for one 
position, the parties' strategies would be determined by the median voter's preference. 
However, for political positions or policies over a multi-dimensional space, pure Nash equilibria 
may never exist~\cite{Dug16}.


\section{Preliminaries}
\label{sec:models}

Let $A$, $B$ be two political parties devoting to an election, such that party $A$ has $m$ candidates 
$A_1,\ldots A_m$ and party $B$ has $n$ candidates $B_1,\ldots,B_n$. Each party has to select one of its candidates to participate in the election. Without loss of generality, $m\geq 2, n\geq 2$. Assume that the society consists of 
the supporters of $A$ and those of $B$. 
Let $u_A(A_i)$ and $u_B(A_i)$ denote the utilities of candidate $A_i$ for party $A$'s supporters and party $B$'s supporters, respectively, for each $i\in [m]$. 
Likewise, let $u_A(B_j)$ and $u_B(B_j)$ denote the utilities of candidate $B_j$ for party $A$'s supporters and party $B$'s supporters, respectively, for each $j\in [n]$. 
Assume that candidates in each party are 
sorted according to the utilities for his or her party's supporters. 
Namely, $u_A(A_1)\geq u_A(A_2)\geq \ldots \geq u_A(A_m)$ and $u_B(B_1)\geq u_B(B_2)\geq \ldots \geq u_B(B_n)$. 
To break the symmetry, assume that $u_A(A_1)\geq u_B(B_1)$. We use $u(A_i)$ (resp., $u(B_j)$) 
to denote the (total) social utility of candidate $A_i$ (resp., $B_j$) for the society, i.e., 
\[u(A_i) = u_A(A_i) + u_B(A_i)\] 
\[(resp., u(B_j) = u_A(B_j) + u_B(B_j)).\] 
Assume that the social utilities are bounded. That is, $u(A_i), u(B_j)\in [0, b]$ for 
some real $b \geq 1$, for each $i\in [m],j\in [n]$. 

Parties $A$ and $B$ represent the two players such that $A$ and $B$ have $m$ and $n$ strategies, respectively. 
A pure strategy of a party is a selected candidate for the election. An important property that we want to preserve in the game is the following: \emph{a party wins the election with higher odds if it selects a candidate with a higher social utility defined above}. 
Inspired by the exploration method used in the multi-armed bandit problem~\cite{Kul00} and the 
probabilistic comparison used in the dueling bandits problem~\cite{AJK2014,YBKJ2012}, 
we formulate the winning odds, $p_{i,j}$, which stands for the probability of~$A_i$ wins against $B_j$, in the following three models to define the payoffs for each party in such a election game. 
\begin{itemize}
\item Linear link model~\cite{AJK2014}: \[p_{i,j} := \frac{1 +(u(A_i)-u(B_j))/b}{2}.\] 
\item Bradley-Terry model~\cite{Bra54,YBKJ2012}: \[p_{i,j} := \frac{u(A_i)}{u(A_i)+u(B_j)}.\] 
\item Softmax model~\cite{Kul00,SB98}\footnote{Actually, we consider a constrained softmax model here. 
	The winning probability is based on the Boltzmann distribution with $kT = b$, where $k$ is 
	the Boltzmann constant and $T$ is the temperature of the system}: 
	\[p_{i,j} := \frac{e^{u(A_i)/b}}{e^{u(A_i)/b}+e^{u(B_j)/b}}.\] 
\end{itemize}

\subsection{Two-Party Election Game}
Now, we are ready to use the concepts introduced above to define the payoffs of our two-party election game.
When the context is clear, we use $(i,j)$ to denote the state of game in which $A_i$ and $B_j$ are selected in the election. The payoff of party $A$ (resp., $B$) in state $(i,j)$, say $a_{i,j}$ 
(resp., $b_{i,j}$), is defined as the expected utilities $A$ (resp., $B$) obtains in state $(i,j)$. 
Namely, 
\begin{eqnarray*}
a_{i,j} &=& p_{i,j}u_A(A_i) + (1-p_{i,j})u_A(B_j)\\
b_{i,j} &=& (1-p_{i,j})u_B(B_j) + p_{i,j}u_B(A_i).
\end{eqnarray*} 
The {\em social welfare} of state $(i,j)$ is $SU_{i,j} = a_{i,j} + b_{i,j}$. We say that 
a state $(i,j)$ is a {\em pure Nash equilibrium} (PNE) if $a_{i',j}\leq a_{i,j}$ for any $i'\neq i$ 
and $b_{i,j'} \leq b_{i,j}$ for any $j'\neq j$. That is, in state $(i,j)$, neither $A$ nor $B$ 
wants to deviate from his or her strategy. The {\em price of anarchy} (PoA) of the game is defined as 
\[
\frac{SU_{i^*,j^*}}{SU_{\hat{i},\hat{j}}} = 
\frac{a_{i^*,j^*}+b_{i^*,j^*}}{a_{\hat{i},\hat{j}}+b_{\hat{i},\hat{j}}},
\]
where $(i^*,j^*) = \arg\max_{(i,j)\in [m]\times [n]}(a_{i,j}+b_{i,j})$ is the {\em optimal state}, 
which has the best (i.e., highest) social welfare among all possible states, and $(\hat{i},\hat{j}) = 
\arg\min_{(i,j)\in [m]\times [n]\atop (i,j)\mbox{\scriptsize \;is a PNE}}(a_{i,j}+b_{i,j})$, 
which is the PNE with the worst (i.e., lowest) social welfare.

The following properties will be needed throughout the paper.

\begin{definition}
We say that the two-party election game is {\em egoistic} if $u_A(A_i)> u_A(B_j)$ and $u_B(B_j)> u_B(A_i)$ for all $i\in [m], j\in [n]$, . 
\end{definition}
This guarantees that \emph{any candidate benefits its supporters more than those from the competing party}, which is reasonable and consistent with the real world.

\begin{definition}
\label{defn:strategy_domination}
For party $X\in \{A, B\}$, we say that strategy $i$ {\em weakly dominates} $i'$ if $i < i'$ and 
$u(X_i)\geq u(X_{i'})$. We say that strategy $i$ {\em dominates} $i'$ if $i$ weakly dominates $i'$ and $u(X_i) > u(X_{i'})$. 
\end{definition}

\begin{remark} 
For all the three models, $p_{i,j}\geq p_{i',j}$ if $u(A_i)\geq u(A_{i'})$.
\end{remark}

\begin{lemma}\label{lem:dominatingNE}
If strategy 1 weakly dominates each $i\in [n]\setminus\{1\}$ in party $A$, then $(1,j^{\#})$ is a PNE of the egoistic two-party 
election game for $j^{\#} = \arg\max_{j\in [m]} b_{1,j}$. Similarly, if strategy 1 weakly dominates each $j\in [m]\setminus \{1\}$ in party $B$, then $(i^{\#}, 1)$ is a PNE for $i^{\#} = \arg\max_{i\in [n]} a_{i,1}$. 
\end{lemma}
\begin{proof}
Suppose that strategy 1 weakly dominates each $i\in [n]\setminus \{1\}$ in party $A$ and $B$ chooses strategy~$j$. 
Clearly, we have $p_{1,j}\geq p_{i,j}$ for any $i\neq 1$. Since 
\begin{eqnarray*}
a_{1,j} - a_{i,j} &=& p_{1,j}u_A(A_1) + (1-p_{1,j})u_A(B_j)\\ &&- \left( p_{i,j}u_A(A_i) + (1-p_{i,j})u_A(B_j)\right)\\
&\geq& u_A(A_1)(p_{1,j}-p_{i,j}) + (p_{i,j} - p_{1,j})u_A(B_j) \\
&=& (p_{1,j} - p_{i,j})(u_A(A_1) - u_A(B_j))\\
&\geq& 0 
\end{eqnarray*}
in which the last inequality follows from the egoistic property, $A$ always wants to choose strategy~1. 
The best response of $B$ is then choose $j^{\#} = \arg\max_{j\in [m]} b_{1,j}$, 
and hence $(1,j^{\#})$ is a PNE. The other case can be analogously proved. 

\end{proof}

Namely, if strategy 1 in $A$ (weakly) dominates each $i\in [n]\setminus \{1\}$ and strategy 1 in $B$ (weakly) dominates each $j\in [m]$, then $(1,1)$ is a (weakly) dominant-strategy solution.

\begin{lemma}\label{lem:fractionals}
Let $r, s > 0$ be two positive real numbers. Then, for any integer $d>0$, $r/s > (r+d)/(s+d)$ if $r > s$ and $r/s < (r+d)/(s+d)$ if $r < s$. 
\end{lemma}
\begin{proof}
Let $r, s > 0$ be two positive integers. Then $r/s - (r+d)/(s+d) = d(r-s)/(s^2+sd)$. So $r/s - (r+d)/(s+d) > 0$ if $r > s$ and $r/s - (r+d)/(s+d) < 0$ if $r < s$. 
\end{proof}

\subsection{No Existence Guarantee of PNE in the Bradley-Terry Model}
\label{subsec:BT}

We claim that the two party election game in the Bradley-Terry model, 
a PNE does not always exist. 
As the upper instance illustrated in Table~\ref{tab:NoPNE_BT}, $m = n = 2$, $u_A(A_1) = 91$, 
$u_B(A_1) = 0$, $u_A(A_2) = 90$, $u_B(A_2) = 8$, $u_B(B_1) = 11$, $u_A(B_1) = 1$, $u_B(B_2) = 10$,
$u_A(B_2) = 20$. We have $p_{1,1} = 91/(91+12)\approx 0.88$, 
$p_{1,2} = 91/(91+30)\approx 0.76$, $p_{2,1} = 98/(98+12)\approx 0.89$, and 
$p_{2,2} = 98/(98+30)\approx 0.77$. 
Hence, we obtain the first matrix in Table~\ref{tab:NoPNE_BT}. Furthermore, as the example is egoistic, it also implies that the game in the Bradley-Terry model may not have a PNE even with egoism guarantee. The second instance in Table~\ref{tab:NoPNE_BT} gives a non-egoistic example of no PNE in the Bradley-Terry model. 

\begin{table}[ht]
\begin{center}
\begin{tabular}[c]{ l l | l l }
	\multicolumn{4}{ c }{}\\
	$A$ & \multicolumn{1}{c}{}& $B$ & \\
	\hline
	$u_A(A_i)$ & $u_B(A_i)$ & $u_B(B_j)$ & $u_A(B_j)$\\
	\hline
	91  &  0  &  11  &  1\\
	90  &  8  &  10  &  20\\
	\hline
\end{tabular} \;\;\;\;
\begin{tabular}[c]{ l l | l l }
	\multicolumn{4}{ c }{}\\
	$A$ & \multicolumn{1}{c}{}& $B$ & \\
	\hline
	$u_A(A_i)$ & $u_B(A_i)$ & $u_B(B_j)$ & $u_A(B_j)$\\
	\hline
	44  &  10  &  37  &  17\\
	39  &  55  &  10  &  5\\
	\hline
\end{tabular}
\vspace{7pt}\\
\begin{tabular}[c]{ l | l }
	\centering
	$a_{1,1}$, $b_{1,1}$  &  $a_{1,2}$, $b_{1,2}$\\
	\hline
	$a_{2,1}$, $b_{2,1}$  &  $a_{2,2}$, $b_{2,2}$\\
\end{tabular}
$\approx$
\begin{tabular}[c]{  l | l  }
	\centering
	80.51, 1.28  &  73.84, 2.17\\
	\hline
	80.29, 8.32  &  74.02, 8.23\\
\end{tabular}
,\;\;
\begin{tabular}[c]{  l | l  }
	\centering
	30.50, 23.50  &  35.52, 10.00\\
	\hline
	30.97, 48.43  &  34.32, 48.81\\
\end{tabular}
\vspace{6pt}
\caption{Two examples of No PNE in the Bradley-Terry model ($m = n = 2, b=100$). Left one is egoistic while the right one is not.}
\label{tab:NoPNE_BT}
\end{center}
\end{table}

\section{Equilibrium Analysis}
\label{sec:m2n2_NE}
We start with the case of two candidates per party, and with the help of Lemma~\ref{lem:dominatingNE} prove existence of PNE in both the linear link model and the softmax model, also using critical lemmas dealing with the complemented condition in Lemma~\ref{lem:dominatingNE}. We then reuse these critical lemmas to show existence of PNE for the case of more than two candidates per party.

\subsection{Two Candidates per Party}

There are exactly two scenarios for the two-party election game with $m = n = 2$ to have no pure Nash equilibrium, as illustrated in Fig.~\ref{fig:2by2_noPNE}. The arrows show the deviations. 
For example, $(D_1)$ stands for $A$'s unilateral deviation from strategy~1 to~2 given $B$ staying at 
strategy~1 and $(D_2)$ stands for $B$'s unilateral deviation from strategy~1 to~2 given $A$ staying at 
strategy~2, while $(D_3)$ stands for $A$'s unilateral deviation from strategy~2 to strategy~1 given $B$ 
staying at strategy~2 and $(D_4)$ stands for $B$'s unilateral deviation from strategy~2 to strategy~1 given $A$ 
staying at strategy~1. 

That is, if the game has no PNE, then at state $(1,1)$, either $A$ or $B$ deviates from strategy~1 to~2. For the former case, since no PNE exists, $B$ wants to deviate unilaterally from 
strategy~1 to~2 at state $(2,1)$, then the game reaches state $(1,2)$, in which $A$ wants to deviate 
unilaterally from strategy~2 to~1. Finally, the entries in the payoff matrix as well as the deviation 
arrows form a cycle as the left scenario of Fig.~\ref{fig:2by2_noPNE} shows. Likewise, the latter 
case corresponds to the right scenario of the Fig.~\ref{fig:2by2_noPNE}.

\begin{figure}[ht]
	\begin{center}
		\includegraphics[scale=0.35]{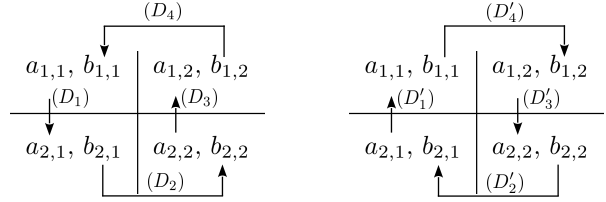}
	\end{center}
	\caption{The two scenarios of having no PNE for $m = n = 2$.}
	\label{fig:2by2_noPNE}
\end{figure}

Let $\Delta(D_i)$ (resp., $\Delta(D'_i)$) denote the gain of payoff by the unilateral deviation $D_i$ 
(resp., $D'_i$) for $i\in \{1,2,3,4\}$. Then, we have 
\begin{eqnarray*}
\Delta(D_1) = -\Delta(D'_1) &=& a_{2,1}-a_{1,1}\\
&=&p_{2,1}u_A(A_2)+(1-p_{2,1})u_A(B_1)\\
&&-(p_{1,1}u_A(A_1)+(1-p_{1,1})u_A(B_1))\\
&=& -p_{1,1}(u_A(A_1)-u_A(A_2))\\ 
&&+ (p_{2,1}-p_{1,1})(u_A(A_2)-u_A(B_1)).\\
\Delta(D_2) = -\Delta(D'_2) &=& b_{2,2}-b_{2,1}\\
&=&(1-p_{2,2})u_B(B_2)+p_{2,2}u_B(A_2)\\ 
&&- ((1-p_{2,1})u_B(B_1)+p_{2,1}u_B(A_2))\\
&=& -(1-p_{2,1})(u_B(B_1)-u_B(B_2))\\
&&+(p_{2,1}-p_{2,2})(u_B(B_2) - u_B(A_2)).
\end{eqnarray*}
Similarly, we derive 
\begin{eqnarray*}
\Delta(D_3) = -\Delta(D'_3) &=& a_{1,2}-a_{2,2}\\ &=&p_{1,2}u_A(A_1)+(1-p_{1,2})u_A(B_2)\\
&&-(p_{2,2}u_A(A_2)+(1-p_{2,2})u_A(B_2))\\
&=& p_{1,2}(u_A(A_1)-u_A(A_2))\\ 
&&+ (p_{1,2}-p_{2,2})(u_A(A_2)-u_A(B_2)).\\
\Delta(D_4) = -\Delta(D'_4) &=& b_{1,1}-b_{1,2}\\ &=&(1-p_{1,1})u_B(B_1)+p_{1,1}u_B(A_1)\\ 
&&- ((1-p_{1,2})u_B(B_2)+p_{1,2}u_B(A_1))\\
&=& (1-p_{1,1})(u_B(B_1)-u_B(B_2))\\
&&+(p_{1,2}-p_{1,1})(u_B(B_2) - u_B(A_1)).
\end{eqnarray*}

In order to check if the two-party election game with $m=n=2$ always has a PNE, 
Lemma~\ref{lem:dominatingNE} suggests us to focus on the case that $u(A_2) > u(A_1)$ and 
$u(B_2) > u(B_1)$ since when either $u(A_1) \geq u(A_2)$ or $u(B_1) \geq u(B_2)$ a PNE always exists.

\subsubsection{Linear Link Model}
\label{subsec:PNE_LL}

In this subsection, we study the linear link model. We show that the two-party election game always has a PNE in this model. 
As previously discussed in this section, we assume that $u(A_2) > u(A_1)$ and 
$u(B_2) > u(B_1)$. The following lemma tells us at least one of the deviations $D_2$ 
and~$D_4$ fails to happen, and analogously, at least one of deviations $D'_1$ and $D'_3$ fails to happen.

\begin{lemma}\label{lem:devConf_LL}
Consider the two-party election game in the linear link model. Suppose that $u(A_2) > u(A_1)$, then $\Delta(D_4) < 0$ (resp., $\Delta(D_2) < 0$) if $\Delta(D_2) > 0$ (resp., $\Delta(D_4) > 0$). Likewise, suppose that $u(B_2) > u(B_1)$, then $\Delta(D'_3) < 0$ (resp., $\Delta(D'_1) < 0$) if $\Delta(D'_1) > 0$ (resp., $\Delta(D'_3) > 0$). 
\end{lemma}
\begin{proof}
	Suppose we have $u(A_2) > u(A_1)$. To ease the notation, let 
	$\hat{p} = 1-p_{2,1}$, $p' = 1-p_{1,1}$, and $\delta = p_{2,1} - p_{2,2} = p_{1,1} - p_{1,2} = (u(B_2)-u(B_1))/2b$. 
	To prove the lemma by contradiction, we assume and focus that $\Delta(D_2)> 0$ yet $\Delta(D_4)\geq 0$, as the other case that $\Delta(D_4)> 0$ yet $\Delta(D_2)\geq 0$ can be proved in the same way. First, we have  
	\[
	\left\{
	\begin{array}{ccc}
	p'(u_B(B_1)-u_B(B_2)) &\geq& \delta\cdot (u_B(B_2)-u_B(A_1)),\\
	\vspace{-5pt}\\
	\hat{p}(u_B(B_1)-u_B(B_2)) &<& \delta\cdot (u_B(B_2)-u_B(A_2)).
	\end{array}
	\right.
	\]
	We claim that $\hat{p}, p'>0$. Recall that $p_{1,1} \geq p_{2,1}$ which implies that $\hat{p}\geq p'$, then $p'= 0$ must be true if one of $\hat{p}$ and $p'$ is~$0$. Substituting $p'=0$ in the first inequality above we have $0\geq \delta\cdot (u_B(B_2)-u_B(A_1))$, which contradicts that  $u_B(B_2) > u_B(A_1)$ by egoism assumption. Dividing the inequalities by $p'$ and $\hat{p}$ respectively, we have 
	\[
	\left\{
	\begin{array}{ccc}
	u_B(B_1)-u_B(B_2) &\geq& \delta\cdot (u_B(B_2)-u_B(A_1))/p',\\
	\vspace{-5pt}\\
	u_B(B_1)-u_B(B_2) &<& \delta\cdot (u_B(B_2)-u_B(A_2))/\hat{p}.
	\end{array}
	\right.
	\]
	Note that $u_B(B_2)> u_B(A_2)$ and $u_B(B_2)> u_B(A_1)$ due to egoism, so $\delta\leq 0$ will make the second inequality fail to hold. Hence we proceed with $\delta > 0$ and compare the right-hand sides of the above inequalities: 
	\begin{eqnarray*}
	&&\frac{(u_B(B_2)-u_B(A_1))/p'}{(u_B(B_2)-u_B(A_2))/\hat{p}}\\ 
	&=& 
	\frac{u_B(B_2)-u_B(A_1)}{u_B(B_2)-u_B(A_2)}\cdot \frac{1+(u(B_1)-u(A_2))/b}{1+(u(B_1)-u(A_1))/b}\\
	&=& \frac{u_B(B_2)-u_B(A_1)}{u_B(B_2)-u_B(A_2)}\cdot \frac{b+(u(B_1)-u(A_2))}{b+(u(B_1)-u(A_1))}. 
	\end{eqnarray*}
Note that $u_B(A_2) > u_B(A_1)$ since $u(A_2) > u(A_1)$ and $u_A(A_1)\geq u_A(A_2)$. Then we have 
	\begin{eqnarray*}
	&&\frac{u_B(B_2)-u_B(A_1)}{u_B(B_2)-u_B(A_2)}\\ 
	&\geq & 
	\frac{u_B(B_2)-u_B(A_1)+(u(B_1)-u_B(B_2))}{u_B(B_2)-u_B(A_2)+(u(B_1)-u_B(B_2))}\quad \mbox{(by Lemma~\ref{lem:fractionals})}\\
	&= & \frac{u(B_1)-u_B(A_1)}{u(B_1)-u_B(A_2)}\\
	&\geq& \frac{u(B_1)-u_B(A_1)+(b-u_A(A_1))}{u(B_1)-u_B(A_2)+(b-u_A(A_2))}\quad \mbox{(by Lemma~\ref{lem:fractionals})} \\
	& =& \frac{u(B_1)-u(A_1)+b}{u(B_1)-u(A_2)+b},  \\
	\end{eqnarray*} 
Finally, we derive that 
\begin{eqnarray*}
&&\frac{(u_B(B_2)-u_B(A_1))/p'}{(u_B(B_2)-u_B(A_2))/\hat{p}}\\
&\geq& 
\frac{u(B_1)-u(A_1)+b}{u(B_1)-u(A_2)+b}\cdot  
\frac{b+(u(B_1)-u(A_2))}{b+(u(B_1)-u(A_1))}= 1,
\end{eqnarray*}
which implies that $u_B(B_1)-u_B(B_2) > u_B(B_1)-u_B(B_2)$, hence a contradiction occurs. For the second part of the lemma that $u(B_2) > u(B_1)$, the proof can be similarly derived. 
	
\end{proof}

Theorem~\ref{thm:PNE_LL} holds by Lemma~\ref{lem:dominatingNE} and~\ref{lem:devConf_LL}.

\begin{theorem}\label{thm:PNE_LL}
In the linear link model with $m = n = 2$, the two-party election game always has a PNE.  
\end{theorem}

\subsubsection{Softmax Model}
\label{subsec:PNE_SM}

In this subsection, we show that the two-party election game always has a PNE of in the softmax model. 

\begin{lemma}\label{lem:devConf_SM}
Consider the two-party election game in the softmax model. Suppose that $u(A_2)> u(A_1)$, then $\Delta(D_4) < 0$ (resp., $\Delta(D_2) < 0$) if $\Delta(D_2)>0$ (resp., $\Delta(D_4) > 0$). Likewise, suppose that $u(B_2)> u(B_1)$, then $\Delta(D'_3) < 0$ (resp., $\Delta(D'_1) < 0$) if $\Delta(D'_1) > 0$ (resp., $\Delta(D'_3) > 0$).  
\end{lemma}
\begin{proof}
    Suppose that $u(B_2)> u(B_1)$. To ease the notation, let 
	\begin{eqnarray*}
	q &=& p_{1,1} = \frac{e^{u(A_1)/b}}{e^{u(A_1)/b} + e^{u(B_1)/b}},\; 
	 q' = p_{2,1} = \frac{e^{u(A_2)/b}}{e^{u(A_2)/b} + e^{u(B_1)/b}},   \;\\ 
	 \delta &=& q' - q.\\
	\hat{q} &=& p_{1,2} = \frac{e^{u(A_1)/b}}{e^{u(A_1)/b} + e^{u(B_2)/b}},\; 
	 \hat{q}'= p_{2,2} = \frac{e^{u(A_2)/b}}{e^{u(A_2)/b} + e^{u(B_2)/b}}, 
	 \;\\ 
	 \delta' &=& \hat{q}' - \hat{q}.
	\end{eqnarray*} 
	To prove the lemma by contradiction, similar to the argument in the proof of Lemma~\ref{lem:devConf_LL}, we focus on the case that $\Delta(D'_3) > 0$ yet $\Delta(D'_1)\geq 0$. The proof for the other case that $\Delta(D'_1)> 0$ yet $\Delta(D'_3)\geq 0$ is basically the same.  
	
	By definition, $q,q',\hat{q},\hat{q'} > 0$. 
	First, as $q,\hat{q}> 0$, $\Delta(D'_3) > 0$ and  $\Delta(D'_1)\geq 0$ implies that   
	\begin{eqnarray*}
		u_A(A_1)-u_A(A_2) &\geq& \delta\cdot \frac{u_A(A_2)-u_A(B_1)}{q},\mbox{ and }\\
		u_A(A_1)-u_A(A_2) &<& \delta'\cdot \frac{u_A(A_2)-u_A(B_2)}{\hat{q}}.
	\end{eqnarray*}
	By the same argument in the proof of Lemma~\ref{lem:devConf_LL}, we focus on $\delta > 0$. Then we have 
		\[
		\left\{
		\begin{array}{llll}
		&&u_A(A_1)-u_A(A_2)\\ 
		&\geq& 
		\frac{e^{u(B_1)/b}(e^{u(A_2)/b} - e^{u(A_1)/b})}{e^{(u(A_1)+u(A_2))/b} + e^{(u(A_1)+u(B_1))/b}}\cdot (u_A(A_2)-u_A(B_1)) 
		& \triangleq \notag{(*)} \\
		\vspace{-7pt}\\
		&&u_A(A_1)-u_A(A_2)\\ 
		&<& \frac{e^{u(B_2)/b}(e^{u(A_2)/b} - e^{u(A_1)/b})}{e^{(u(A_1)+u(A_2))/b}+e^{(u(A_1)+u(B_2))/b}}\cdot (u_A(A_2)-u_A(B_2)) 
		& \triangleq \notag{(**)} 
		\end{array} 
		\right.
		\]
	Then, we compare the right-hand sides of the above two 
	inequalities as below.  
	\begin{eqnarray*}
		\frac{(*)}{(**)} &=& \frac{e^{u(B_1)/b}}{e^{u(B_2)/b}}\cdot 
		\frac{e^{u(A_2)/b}+e^{u(B_2)/b}}{e^{u(A_2)/b}+e^{u(B_1)/b}}\cdot 
		\frac{u_A(A_2)-u_A(B_1)}{u_A(A_2)-u_A(B_2)} \\
		&\geq& e^{(u(B_1)-u(B_2))/b}\cdot \frac{u_A(A_2)-u_A(B_1)}{u_A(A_2)-u_A(B_2)}. 
	\end{eqnarray*}
	Let $c = u_A(A_2) - u_A(B_2)$ and $\rho = u(B_2)-u(B_1) >0$. Then, 
	$u_A(A_2) - u_A(B_1) = c+u_A(B_2)-u_A(B_1)\geq c+(u(B_2)-u(B_1)) = c+\rho$. 
	Note that $e^{(u(B_1)-u(B_2))/b}\geq 
	1 + (u(B_1)-u(B_2))/b = 1 - \rho/b$. 	
	Also, $(u_A(A_2)-u_A(B_1))/(u_A(A_2)-u_A(B_2))\geq (c+\rho)/c = 1+\rho/c$. 
	Therefore, 
	\begin{eqnarray*}
		\frac{(*)}{(**)} &\geq & 
		\left(1-\frac{\rho}{b}\right)\cdot \left(1+\frac{\rho}{c}\right)  \\
		& = & 1 + \frac{\rho}{c} - \frac{\rho}{b} -\frac{\rho^2}{bc}\\
		& = & 1 + \rho\cdot \frac{b-c - \rho}{bc} \\
		&\geq& 1.  
	\end{eqnarray*} 
	where the last inequality follows because 
	$b-c-\rho = b - u_A(A_2) + u_A(B_2) - u_B(B_2) - u_A(B_2) + u_B(B_1) + u_A(B_1) 
	= (b-u_A(A_2)) + (u_B(B_1)-u_B(B_2))+u_A(B_1)\geq 0$. 
	Finally, we obtain a contradiction. The other part of the lemma that $u(B_2)> u(B_1)$ can be similarly proved.  
	
\end{proof}

By Lemma~\ref{lem:dominatingNE} and~\ref{lem:devConf_SM}, we obtain Theorem~\ref{thm:PNE_SM} 
as follows.

\begin{theorem}\label{thm:PNE_SM}
In the softmax model with $m = n = 2$, the two-party election game always has a PNE. 
\end{theorem}

\subsection{More than Two Candidates per Party}
\label{sec:general_NE}

In this section, we show existence of PNE for a generalized case in which party $A$ has $m\geq 2$ 
candidates and party $B$ has $n\geq 2$ candidates. Since the game in the Bradley-Terry model may 
not have a PNE, we focus on the linear link model and the softmax model.

The two-party election game then has $mn$ possible states (i.e., $\{(i,j)\}_{i\in [m],j\in [n]}$).
The states are bijectively mapped to the entries of the payoff matrix. Regard each entry of the 
payoff matrix as a node and a unilateral deviation with positive gain of payoff as an arc, we 
obtain a directed graph and we call it the {\em state graph}. A {\em best-response walk} is a walk 
on the state graph such that each arc of the walk is a best-response unilateral deviation. Since 
the number of states is finite, any best-response walk on the state graph must contain a loop if the 
game has no PNE. We then derive Theorem~\ref{thm:generel_NE} based on this observation.


\begin{figure}[ht]
	\begin{center}
		\includegraphics[scale=0.40]{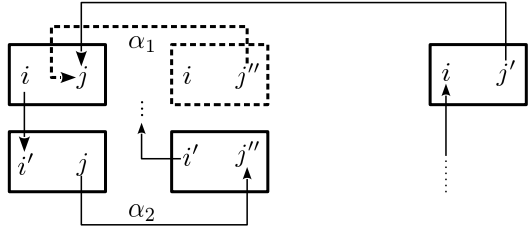}
	\end{center}
\caption{Illustration for the proof of Theorem~\ref{thm:generel_NE}.}
\label{fig:general_NE}
\end{figure}

\begin{theorem}\label{thm:generel_NE}
The two-party election game with $m\geq 2$ and $n\geq 2$ always has a PNE both in the linear 
link model and the softmax model.
\end{theorem}
\begin{proof}
Assume, for contradiction, that the game has no PNE. Then, a best-response walk starting 
from any node contains a loop. Let $L = (i,j)\rightarrow (i',j)\rightarrow (i',j'')\rightarrow\cdots 
\rightarrow (i,j')\rightarrow (i,j)$ be such a loop, 
in which $i$ is the smallest indexed party $A$ candidate among all nodes (states) of~$L$ 
(See Fig.~\ref{fig:general_NE} for the illustration). 
Let $\alpha_1$ and $\alpha_2$ denote the unilateral deviations $(i,j'')\rightarrow (i,j)$ and 
$(i',j)\rightarrow (i',j'')$ respectively. Since $\alpha_2$ is on~$L$, we have $\Delta(\alpha_2)>0$.  
Since $(i,j')\rightarrow (i,j)$ is the best response of $B$ when $A$ selects $i$, 
it must be that $b_{i,j} \geq b_{i,j''}$. Thus, if $A$ selects $i$ and $B$ selects $j''$, 
then we must have $\Delta(\alpha_1)\geq 0$. 

Moreover, since $u(A_i)\geq u(A_{i'})$ implies $a_{i,j}\geq a_{i',j}$, which is impossible 
because $(i,j)\rightarrow(i',j)$ is on $L$, we have $u(A_{i'}) > u(A_i)$. 
Then by Lemma~\ref{lem:devConf_LL} and~\ref{lem:devConf_SM} (with $(i,j)$, $(i',j)$, $(i',j'')$, 
and $(i,j'')$ corresponding to $(1,1)$, $(2,1)$, $(2,2)$, and $(1,2)$ in the case for two candidates, 
respectively), we know that $\Delta(\alpha_2)>0$ implies that $\Delta(\alpha_1) < 0$, 
hence a contradiction is derived.
\end{proof}

\section{Price of Anarchy}
\label{sec:PoA}

A PNE may give a state that is suboptimal in the social welfare for this two-party system. We show a tight bound of $2$ on the PoA in the linear link model and the Bradley-Terry model. At the end of this section, we give an upper bound of $1+e$ in the softmax model.

Below is a useful claim which is used in this section. 
 
\begin{lemma}\label{cla:dominated_Not_PNE}
If $i$ is dominated by some $i'$, then $(i,j)$ is not a PNE for any $j\in [n]$. Similarly, 
if $j$ is dominated by some $j'$, then $(i,j)$ is not a PNE for any $i\in [m]$.  
\end{lemma}
\begin{proof}
If $i$ is dominated by $i'$, then $i'< i$ and $u(A_{i'}) > u(A_i)$. 
So $u_A(A_{i'}) > u_A(A_i)$ and $p_{i',j} > p_{i,j}$ for any~$j\in [n]$. Thus, 
$a_{i,j} - a_{i',j} = p_{i,j}u_A(A_i)+(1-p_{i,j})u_A(B_j) - (p_{i',j}u_A(A_{i'})+(1-p_{i',j})u_A(B_j))
< (p_{i,j} - p_{i',j})(u_A(A_i) - u_A(B_j)) < 0$, which implies that $A$ wants to 
deviate from~$i$ to~$i'$. Then $(i,j)$ is not a PNE. The second part of the claim can 
be similarly proved. 
\end{proof}

\subsection{Linear Link Model}
\label{subsec:PoA_LL}

In this subsection, we show that, in the linear link model, the PoA of the two-party election game is bounded by~$2$. We also give an instance to justify that the bound is tight.


Note that 
\begin{eqnarray*}
	SU_{i,j} &=& p_{i,j}\cdot u(A_i) + (1-p_{i,j})\cdot u(B_j)\\
	&=&\frac{1+(u(A_i)-u(B_j))/b}{2} \cdot u(A_i)\\ 
	&&+ \frac{1-(u(A_i)-u(B_j))/b}{2}\cdot u(B_j)\\
	&=& \frac{1}{2}(u(A_i)+u(B_j)) + \frac{1}{2b}(u(A_i)-u(B_j))^2\\
	&\geq& \frac{1}{2}(u(A_i)+u(B_j)).
\end{eqnarray*}
and 
\begin{eqnarray*}
	SU_{i,j} &=& p_{i,j}\cdot u(A_i) + (1-p_{i,j})\cdot u(B_j)\leq \max\{u(A_i), u(B_j)\}.
\end{eqnarray*}

\begin{theorem}\label{thm:PoA_LL}
	The two-party election game in the linear link model 
	has the price of anarchy bounded by~2. 
\end{theorem}
\begin{proof}
Let $(i,j)$ be a PNE and $(i^*,j^*)$ be the optimal state. So we have, 
\[
\left\{\begin{array}{l}
i\mbox{ is not dominated by } i^*\\
j\mbox{ is not dominated by } j^*
\end{array}
\right.
\Rightarrow
\left\{\begin{array}{l}
i\leq i^*\mbox{ or } u(A_{i^*})\leq u(A_i)\\
j\leq j^*\mbox{ or } u(B_{j^*})\leq u(B_j)
\end{array}
\right.
\]
Recall that $SU_{i^*,j^*}\leq \max\{u(A_{i^*}), u(B_{j^*})\}$. Moreover, 
$2\cdot SU_{i,j}\geq u(A_i)+u(B_j)$. Consider the following four cases.
\begin{enumerate}
	\item If $i\leq i^*$ and $j\leq j^*$, then, 
	$u(A_i)+u(B_j) = u_A(A_i)+u_B(A_i)+u_B(B_j)+u_A(B_j)\geq u_A(A_{i^*})+u_B(A_{i^*}) = u(A_{i^*})$ 
	since $u_A(A_i)\geq u_A(A_{i^*})$ and $u_B(B_j)> u_B(A_{i^*})$. 
	Similarly, we have $u(A_i)+u(B_j)\geq u(B_{j^*})$. 
	Hence, we have $u(A_i)+u(B_j)\geq \max\{u(A_{i^*}), u(B_{j^*})\}$. 
	\item If $i\leq i^*$ and $u(B_{j^*})\leq u(B_j)$, then 
	$u(A_i)+u(B_j)\geq u(A_{i^*})$ and $u(A_i) + u(B_j)\geq u(B_j)\geq u(B_{j^*})$ 
	obviously holds. So, we have $u(A_i)+u(B_j)\geq \max\{u(A_{i^*}), u(B_{j^*})\}$. 
	\item If $u(A_{i^*})\leq u(A_i)$ and $j\leq j^*$, similarly we obtain 
	$u(A_i)+u(B_j)\geq \max\{u(A_{i^*}), u(B_{j^*})\}$. 
	\item If $u(A_{i^*})\leq u(A_i)$ and $u(B_{j^*}) \leq u(B_j)$, obviously we have 
	$u(A_i)+u(B_j)\geq u(A_{i^*}) + u(B_{j^*})\geq \max\{u(A_{i^*}), u(B_{j^*})\}$. 
\end{enumerate}
Thus, we conclude that $SU_{i,j}\geq SU_{i^*,j^*}/2$, Therefore, the PoA is bounded by~2. 
\end{proof}

\subsubsection{A Tight Example.}
\label{subsubsec:LB_LL}

Let us consider the game instance of $m=n=2$ in Table~\ref{tab:PoA_tight}. Let $\epsilon,\delta > 0$ 
be two small constants such that $\epsilon,\delta \ll b$ and $\delta\ll \epsilon$. In the linear link model, 
$p_{1,1} = p_{2,2} = 1/2$, $p_{1,2} = (1-(\epsilon-2\delta)/b)/2 = 1/2 - (\epsilon-2\delta)/2b$, 
$p_{2,1} = (1+(\epsilon-2\delta)/b)/2 = 1/2 + (\epsilon-2\delta)/2b$. Hence, 
\[
\left\{
\begin{array}{l}
a_{1,1} = b_{1,1} = \frac{\epsilon}{2},\\
a_{1,2} = b_{2,1} = 
\epsilon - \delta\cdot(\frac{1}{2}+ \frac{\epsilon-2\delta}{2b})\approx 
\epsilon - \frac{\delta}{2},\\
a_{2,1} = b_{1,2} = 
\frac{\epsilon}{2} + \frac{1}{2b}\cdot ((\epsilon-2\delta)(\epsilon-\delta) - b\delta) \approx 
\frac{\epsilon}{2} - \frac{\delta}{2}.
\end{array}\right.
\]
It is easy to see that $(1,1)$ is the only PNE and $(2,2)$ is the optimal state. Thus, the PoA of the instance is approximately $$2-\frac{2\delta}{\epsilon},$$ which is close to~2 
as $\delta/\epsilon$ approaches~0.

\begin{table}[ht]
\begin{center}
	\begin{tabular}[c]{ l l | l l }
		\multicolumn{4}{ c }{}\\
		$A$ & \multicolumn{1}{c}{}& $B$ & \\
		\hline
		$u_A(A_i)$ & $u_B(A_i)$ & $u_B(B_j)$ & $u_A(B_j)$\\
		\hline
		$\epsilon$  &  0 &  $\epsilon$  &  0\\
		$\epsilon-\delta$   &  $\epsilon-\delta$  &  $\epsilon-\delta$   &  $\epsilon-\delta$\\
		\hline
	\end{tabular}
	\vspace{7pt}\\
	\renewcommand{\arraystretch}{1.25}
	\begin{tabular}[c]{ l | l }
		\centering
		$a_{1,1}$, $b_{1,1}$  &  $a_{1,2}$, $b_{1,2}$\\
		\hline
		$a_{2,1}$, $b_{2,1}$  &  $a_{2,2}$, $b_{2,2}$\\
	\end{tabular}
	$\approx$
	\renewcommand{\arraystretch}{1.25}
	\begin{tabular}[c]{  l l | l l }
		\centering
		$\frac{\epsilon}{2}$, & $\frac{\epsilon}{2}$  &  
		$\epsilon-\frac{\delta}{2}$, & $\frac{\epsilon}{2}-\frac{\delta}{2}$\\[2pt]
		\hline 
		$\frac{\epsilon}{2}-\frac{\delta}{2}$, &  $\epsilon-\frac{\delta}{2}$  &  
		$\epsilon-\delta$, & $\epsilon-\delta$\\
	\end{tabular}
	\vspace{6pt}
	\caption{An example illustrating a lower bound on the worst PoA of the egoistic game in both the linear link model and the softmax model ($m = n = 2$, $0<\epsilon, \delta \ll b$, and $\delta \ll \epsilon$).}
	\label{tab:PoA_tight}
\end{center}
\end{table}

\subsection{Bradley-Terry Model}
\label{subsec:PoA_BT}

Though in Sect.~\ref{subsec:BT}, we show that a PNE does not always exist in the Bradley-Terry model, one may be curious about how good or bad its PoA is once it has a PNE. In this subsection, we show that its price of anarchy is bounded by~$2$ in this model once a PNE exists in the game. A lower bound on the PoA is also given.

\begin{theorem}\label{thm:PoA_BT}
	The two-party election game in the Bradley-Terry model 
	has the price of anarchy bounded by~2. 
\end{theorem}
\begin{proof}
Recall that $SU_{i^*,j^*}\leq \max\{u(A_{i^*}), u(B_{j^*})\}$. 
Moreover, the Cauchy-Schwarz inequality implies that 
$u(A_i)^2+u(B_j)^2\geq (u(A_i)+u(B_j))^2/2$. 
Hence, we obtain that
\begin{eqnarray*}
SU_{i,j}&=& \frac{u(A_i)}{u(A_i)+u(B_j)}\cdot u(A_i) + 
\frac{u(B_j)}{u(A_i)+u(A_j)}\cdot u(B_j)\\ 
&=&  \frac{u(A_i)^2+u(B_j)^2}{u(A_i)+u(B_j)} \geq \frac{1}{2}\cdot (u(A_i)+u(B_j)).
\end{eqnarray*}
Refer to the four cases in the proof of Theorem~\ref{thm:PoA_LL}, 
the rest of the proof follows.  
\end{proof}

\subsubsection{A Lower Bound Example.}

Consider the instance in Table~\ref{tab:PoA_LB_BT}. We have $p_{1,1} = p_{2,2} = 1/2$, 
$p_{1,2} = \epsilon/(2\epsilon-\delta)$, $p_{2,1} = (\epsilon-\delta)/(2\epsilon-\delta)$. 
Hence, $a_{1,1} = b_{1,1} = \epsilon/2$, $a_{2,1} = b_{1,2}\approx (\epsilon-\delta)/2$, 
and $a_{2,2} = b_{2,2}\approx \epsilon/2$. Then states $(1,1)$ and $(2,2)$ are both 
PNE and the PoA is bounded by~$(3\epsilon/2-\delta/2)/\epsilon<3/2$.

\begin{table}[ht]
\begin{center}
	\begin{tabular}[c]{ l l | l l }
		\multicolumn{4}{ c }{}\\
		$A$ & \multicolumn{1}{c}{}& $B$ & \\
		\hline
		$u_A(A_i)$ & $u_B(A_i)$ & $u_B(B_j)$ & $u_A(B_j)$\\
		\hline
		$\epsilon$  &  0 &  $\epsilon$  &  0\\
		$\epsilon-\delta$   &  $\delta$  &  $\epsilon-\delta$   &  $\delta$\\
		\hline
	\end{tabular}
	\vspace{7pt}\\
	\renewcommand{\arraystretch}{1.25}
	\begin{tabular}[c]{ l | l }
		\centering
		$a_{1,1}$, $b_{1,1}$  &  $a_{1,2}$, $b_{1,2}$\\
		\hline
		$a_{2,1}$, $b_{2,1}$  &  $a_{2,2}$, $b_{2,2}$\\
	\end{tabular}
	$\approx$
	\renewcommand{\arraystretch}{1.25}
	\begin{tabular}[c]{  l l | l l }
		\centering
		$\frac{\epsilon}{2}$, & $\frac{\epsilon}{2}$  &  
		$\epsilon$, & $\frac{\epsilon-\delta}{2}$\\[2pt]
		\hline 
		$\frac{\epsilon-\delta}{2}$, &  $\epsilon$  &  
		$\frac{\epsilon}{2}$, & $\frac{\epsilon}{2}$ \\
	\end{tabular}
	\vspace{6pt}
	\caption{An example illustrating a lower bound $3/2$ on the worst PoA of the egoistic game in the 
		Bradley-Terry model ($m = n = 2$, $0<\delta \ll \epsilon$).}
	\label{tab:PoA_LB_BT}
\end{center}
\end{table}

\subsection{Softmax Model}
\label{subsec:PoA_SM}

In this subsection, we show that the price of anarchy of the game 
is bounded by~$1+e$ in the softmax model. A lower bound on the PoA is also given.


Note that $p_{i,j}\geq e^0/(e^0+e)\geq 1/(1+e)$ and $p_{i,j}\leq e/(e^0+e)\leq e/(1+e)$ for any $i,j$. 
\begin{eqnarray*}
	SU_{i,j} &=& p_{i,j}\cdot u(A_i) + (1-p_{i,j})\cdot u(B_{j})\\
	&\geq& \min\{p_{i,j}, 1-p_{i,j}\}\cdot (u(A_i)+u(B_j))\\
	&\geq & \frac{1}{1+e}(u(A_i)+u(B_j)), 
\end{eqnarray*}
and 
\begin{eqnarray*}
	SU_{i,j} &=& p_{i,j}\cdot u(A_i) + (1-p_{i,j})\cdot u(B_{j})\\
	&\leq& 
	\frac{e}{1+e}\cdot \max\{u(A_i),u(B_j)\} + \frac{1}{1+e}\cdot \min\{u(A_i),u(B_j)\}.
\end{eqnarray*}
Now, we are ready for Theorem~\ref{thm:PoA_SF} and its proof. 

\begin{theorem}\label{thm:PoA_SF}
	The two-party election game in the softmax model has the PoA bounded by~$1+e$. 
\end{theorem}
\begin{proof}
Let $(i,j)$ be a PNE and $(i^*,j^*)$ be the optimal state. So we have, 
\[
\left\{\begin{array}{l}
i\mbox{ is not dominated by } i^*\\
j\mbox{ is not dominated by } j^*
\end{array}
\right.
\Rightarrow
\left\{\begin{array}{l}
i\leq i^*\mbox{ or } u(A_{i^*})\leq u(A_i)\\
j\leq j^*\mbox{ or } u(B_{j^*})\leq u(B_j)
\end{array}
\right.
\]
Recall that $(1+e)\cdot SU(i,j)\geq (u(A_i)+u(B_j))$. Without loss of generality, let $\max\{u(A_{i^*}),u(B_{j^*})\} = u(A_{i^*})$. 
Consider the following four cases. 
\begin{enumerate}
	\item $i\leq i^*$ and $j\leq j^*$. Then, 
	\[
	u(A_i)+u(B_j) - (\frac{e}{1+e}\cdot u(A_{i^*}) + \frac{1}{1+e}\cdot u(B_{j^*}))\geq 0
	\] since $u_A(A_i)\geq (e/(1+e))u_A(A_{i^*}) + (1/(1+e))u_A(B_{j^*})$ and $u_B(B_j)\geq (e/(1+e))u_B(A_{i^*}) + \break (1/(1+e))u_B(B_{j^*})$. 
	\item $i\leq i^*$ and $u(B_{j^*})\leq u(B_j)$. Then 
	\begin{eqnarray*}
	& & u(A_i)+u(B_j) - (\frac{e}{1+e}\cdot u(A_{i^*}) + \frac{1}{1+e}\cdot u(B_{j^*}))\\
	& \geq & (u_A(A_i)+u_B(A_i))+ \frac{e}{1+e}(u_B(B_j)+u_A(B_j))\\ 
	&&- \frac{e}{1+e}(u_A(A_{i^*}) + u_B(A_{i^*}))\\
	&\geq & 0, 
	\end{eqnarray*}
	in which the last inequality holds since $u_A(A_i)\geq u_A(A_{i^*})$ and $u_B(B_j) > u_B(A_{i^*})$. 
	\item $u(A_{i^*})\leq u(A_i)$ and $j\leq j^*$. This case is similar to (2). 
	\item $u(A_{i^*})\leq u(A_i)$ and $u(B_{j^*})\leq u(B_j)$. Obviously, 
	\begin{eqnarray*}
	&&u(A_i)+u(B_j) - (\frac{e}{1+e}\cdot u(A_{i^*}) + \frac{1}{1+e}\cdot u(B_{j^*})) \\
	&\geq& \frac{1}{1+e}\cdot u(A_{i})+ \frac{e}{1+e}\cdot u(B_{j}) \geq 0.
	\end{eqnarray*}
\end{enumerate}
Therefore, we conclude that $SU_{i,j}\geq SU_{i^*,j^*}/(1+e)$, Therefore, the PoA is bounded by~$1+e$. 
\end{proof}

\subsubsection{A Lower Bound Example.}

Consider the instance in Table~\ref{tab:PoA_tight}. We have $p_{1,1} = p_{2,2} = 1/2$, $p_{1,2} = e^{\epsilon} / (e^{\epsilon}+e^{2\epsilon-2\delta})\approx 1/2$, 
$p_{2,1} = e^{2\epsilon-2\delta}/(e^{\epsilon}+e^{2\epsilon-2\delta})\approx 1/2$. 
Hence, $a_{1,1} = b_{1,1} = \epsilon/2$, $a_{2,1} = b_{1,2}\approx (\epsilon-\delta)/2$, 
and $a_{2,2} = b_{2,2}  = \epsilon - \delta$. Similar to the analysis in Sect.~\ref{subsubsec:LB_LL}, 
we obtain that the price of anarchy of this instance is approximately $2-\frac{2\delta}{\epsilon}$, 
which is close to~2 as $\delta/\epsilon$ approaches~0.

\section{Non-Egoistic Game}


Without the egoistic property, a game instance with no PNE in the linear link model 
can be constructed, and another game instance with an unbounded PoA can also be given in the three models.
\subsection{No Existence Guarantee of PNE in the Linear Link Model} \label{app:equil}


Consider the two-party election game in the linear link model. As the instance illustrated in 
Table~\ref{tab:NoPNE_LL_no_ego}, it is \emph{not} an egoistic one (e.g., $u_B(A_1)> u_B(B_2)$). We derive that 
$p_{1,1} = (1+(60-100)/100)/2 = 0.3$, $p_{1,2} = (1+(60-25)/100)/2 = 0.675$, 
$p_{2,1} = (1+(25-100)/100)/2 = 0.125$, and $p_{2,2} = (1+(25-25)/100)/2 = 0.5$. Hence, we obtain 
the payoff matrix as illustrated in the bottom of Table~\ref{tab:NoPNE_BT}.

\begin{table}[ht]
\begin{center}
	\begin{tabular}[c]{ l l | l l }
		\multicolumn{4}{ c }{}\\
		$A$ & \multicolumn{1}{c}{}& $B$ & \\
		\hline
		$u_A(A_i)$ & $u_B(A_i)$ & $u_B(B_j)$ & $u_A(B_j)$\\
		\hline
		50  &  10  &  10  &  90\\
		5   &  20  &  5   &  20\\
		\hline
	\end{tabular}
	\vspace{7pt}\\
	\begin{tabular}[c]{ l | l }
		\centering
		$a_{1,1}$, $b_{1,1}$  &  $a_{1,2}$, $b_{1,2}$\\
		\hline
		$a_{2,1}$, $b_{2,1}$  &  $a_{2,2}$, $b_{2,2}$\\
	\end{tabular}
	$=$
	\begin{tabular}[c]{  l l | l l }
		\centering
		78, & 10  &  40.25, & 8.375\\
		\hline
		79.375, & 11.25  &  12.5, & 12.5\\
	\end{tabular}
	\vspace{6pt}
\caption{An example having no PNE in the linear link model ($m = n = 2$, $b = 100$).}
\label{tab:NoPNE_LL_no_ego}
\end{center}
\end{table}

From the bottom of Table~\ref{tab:NoPNE_LL_no_ego}, none of the states is a PNE (e.g., in state $(1,1)$, $A$ wants to deviate from his strategy to 2 because he or she will get the payoff 79.375 which is better 
than~78).

\subsection{Unbounded PoA in the Three Models} \label{app:poa}

In this subsection, we show that the two-party election game has unbounded 
PoA if it is \emph{not} egoistic. Let us consider the game instance illustrated in 
Table~\ref{tab:BadPoA_all_NoEgo}. Its payoff matrix is different and we will 
show that its PoA is unbounded in the linear link, softmax, and Bradley-Terry models. 

\begin{table}[ht]
\begin{center}
	\begin{tabular}[c]{ l l | l l }
		\multicolumn{4}{ c }{}\\
		$A$ & \multicolumn{1}{c}{}& $B$ & \\
		\hline
		$u_A(A_i)$ & $u_B(A_i)$ & $u_B(B_j)$ & $u_A(B_j)$\\
		\hline
		$\epsilon$  &  0  &  $\epsilon$  &  0\\
		0  &  $b$  &  0  &  $b$\\
		\hline
	\end{tabular}
\end{center}
\caption{An illustrating non-egoistic game instance for $m = n = 2$.}
\label{tab:BadPoA_all_NoEgo}
\end{table}

\subsubsection{The Linear Link Model}

Let $\epsilon > 0$ be a small constant. 
We derive that $p_{1,1} = p_{2,2} = 1/2$, 
$p_{1,2} = (1+(\epsilon-b)/b)/2$, and $p_{2,1} = (1+(b-\epsilon)/b)/2$. 
Hence, we obtain the payoff matrix as illustrated in Table~\ref{tab:BadPoA_LL_NoEgo}. 
Clearly, state $(1,1)$ is a PNE, so the PoA is at least $b/\epsilon$,  
which is unbounded as $\epsilon$ approaches to~0.

\begin{table}[ht]
\begin{center}
	\renewcommand{\arraystretch}{1.25}
	\begin{tabular}[c]{ l | l }
		\centering
		$a_{1,1}$, $b_{1,1}$  &  $a_{1,2}$, $b_{1,2}$\\
		\hline
		$a_{2,1}$, $b_{2,1}$  &  $a_{2,2}$, $b_{2,2}$\\
	\end{tabular}
	$=$
	\renewcommand{\arraystretch}{1.25}
\begin{tabular}[c]{ c c | c c }
	\centering
	$\frac{\epsilon}{2}$, & $\frac{\epsilon}{2}$  &  
	$b-\frac{\epsilon(b-\epsilon)}{2b}$, & 0\\[2pt]
	\hline
	0, & $b-\frac{\epsilon(b-\epsilon)}{2b}$  &  $\frac{b}{2}$, & $\frac{b}{2}$\\[3pt]
\end{tabular}
	\vspace{6pt}
	\caption{The payoff matrix of instance in Table~\ref{tab:BadPoA_all_NoEgo} in the linear link model.}
	\label{tab:BadPoA_LL_NoEgo}
\end{center}	
\end{table}

\subsubsection{The Softmax Model}

By definition of the softmax model, we derive that $p_{1,1} = p_{2,2} = 1/2$, 
$p_{1,2} = e^{\epsilon/b}/(e^{\epsilon/b} + e)$, and $p_{2,1} = e/(e^{\epsilon/b} + e)$. 
Hence, we obtain the payoff matrix as illustrated in the bottom of Table~\ref{tab:BadPoA_SM_NoEgo}. 
Clearly, state $(1,1)$ is a PNE, so the PoA is at least 
\[
\frac{b}{{2\epsilon e^{\epsilon}/(e^{\epsilon}+1)}},
\]
which is unbounded as $\epsilon$ approaches to~0.

\begin{table}[ht]
\begin{center}
	\renewcommand{\arraystretch}{1.25}
	\begin{tabular}[c]{ l | l }
		\centering
		$a_{1,1}$, $b_{1,1}$  &  $a_{1,2}$, $b_{1,2}$\\
		\hline
		$a_{2,1}$, $b_{2,1}$  &  $a_{2,2}$, $b_{2,2}$\\
	\end{tabular}
	$=$
	\renewcommand{\arraystretch}{1.25}
\begin{tabular}[c]{ c c | c c }
	\centering
	$\frac{\epsilon e^{\epsilon}}{e^{\epsilon}+1}$, & $\frac{\epsilon e^{\epsilon}}{e^{\epsilon}+1}$  &  
	$\frac{\epsilon e^{\epsilon}+eb}{e^{\epsilon}+1}$, & 0\\[2pt]
	\hline
	0, & $\frac{\epsilon e^{\epsilon}+eb}{e^{\epsilon}+1}$  &  $\frac{b}{2}$, & $\frac{b}{2}$\\[3pt]
\end{tabular}
	\vspace{6pt}
	\caption{The payoff matrix of instance in Table~\ref{tab:BadPoA_all_NoEgo} in the softmax model.}
	\label{tab:BadPoA_SM_NoEgo}
\end{center}	
\end{table}

\subsubsection{The Bradley-Terry Model}

By definition of the Bradley-Terry model, we derive 
that $p_{1,1} = p_{2,2} = 1/2$, $p_{1,2} = \epsilon/(\epsilon + b)$, 
and $p_{2,1} = b/(\epsilon +b)$. 
Hence, we obtain the payoff matrix as illustrated in the bottom of Table~\ref{tab:BadPoA_BT_NoEgo}. 
Clearly, state $(1,1)$ is a PNE, so the PoA is at least $b/\epsilon$,  
which is unbounded as $\epsilon$ approaches to~0.

\begin{table}[ht]
\begin{center}
	\renewcommand{\arraystretch}{1.25}
	\begin{tabular}[c]{ l | l }
		\centering
		$a_{1,1}$, $b_{1,1}$  &  $a_{1,2}$, $b_{1,2}$\\
		\hline
		$a_{2,1}$, $b_{2,1}$  &  $a_{2,2}$, $b_{2,2}$\\
	\end{tabular}
	$=$
	\renewcommand{\arraystretch}{1.25}
\begin{tabular}[c]{ c c | c c }
	\centering
	$\frac{\epsilon}{2}$, & $\frac{\epsilon}{2}$  &  
	$\frac{\epsilon^2 + b^2}{b+\epsilon}$, & 0\\[2pt]
	\hline
	0, & $\frac{\epsilon^2 + b^2}{b+\epsilon}$  &  $\frac{b}{2}$, & $\frac{b}{2}$\\[3pt]
\end{tabular}
	\vspace{6pt}
	\caption{The payoff matrix of instance in Table~\ref{tab:BadPoA_all_NoEgo} in the Bradley-Terry model.}
	\label{tab:BadPoA_BT_NoEgo}
\end{center}	
\end{table}

\section{Conclusions and Future Work}
\label{sec:future}

We summarize our results in Table~\ref{tab:summary}.

\begin{table}[ht]
\begin{center}
	\renewcommand{\arraystretch}{1.25}
	\begin{tabular}{ r|c|c|c| }
		 \multicolumn{1}{r}{}
		&  \multicolumn{1}{r}{Linear Link}
		&  \multicolumn{1}{r}{Bradley-Terry}
		&  \multicolumn{1}{r}{Softmax}\\
		\cline{2-4}
		PNE w/ egoism & \checkmark & $\times$ & \checkmark \\
		PNE w/o egoism & $\times$ & $\times$ & ?$^{\#}$  \\
		\cline{2-4}
		\hline
		\hline
		Worst PoA w/ egoism & $\leq 2^*$ & $\leq 2$ & $\leq 1+e$ \\
		Worst PoA w/o egoism & $\infty$ & $\infty$ & $\infty$ \\
		\cline{2-4}
	\end{tabular}
	\vspace{7pt}
\caption{$^*$: the bound is tight. $\infty$: unbounded. \checkmark: PNE always exists. 
$\times$: PNE does NOT always exist. $^{\#}$: Based on our simulation results, we conjecture that the game in the softmax model always has a PNE (even without egoism).}
\label{tab:summary}
\end{center}
\end{table}

In this paper, we only focus on pure Nash equilibria. The other equilibrium concepts, 
such as mixed Nash equilibria, approximate Nash equilibria, etc., also deserve 
further investigations. We conjecture that the two-party election game is not a 
smooth game, and different equilibrium concepts may have different tight bounds on 
the corresponding price of anarchy.

There is a still a gap between the upper bound and lower bound on the worst PoA of the egoistic 
two-party election game in the softmax model and the Bradley-Terry model as well. By examining the game instances randomly sampled, we conjecture that the upper bound is at most~2 in the softmax model and strictly below~2 in the Bradley-Terry model. 

It will be interesting to generalize our results to election games of two or more parties. 
Coalition between the parties can be also considered. Moreover, to design mechanisms toward 
settings of the linear link model or the softmax model is another interesting direction.

\end{document}